\documentclass[letterpaper,10 pt, conference]{ieeeconf}

\newif\ifarxiv
\arxivtrue

\newif\ifreview
\reviewfalse

\input{Sources/header.tex}

\title{\textbf{Closed-loop Performance Optimization of Model Predictive Control with Robustness Guarantees}}
\author{Riccardo Zuliani, Efe C. Balta, and John Lygeros%
\thanks{Corresponding author: R. Zuliani. This work was supported as a part of NCCR Automation, a National Centre of Competence in Research, funded by the Swiss National Science Foundation (grant number 51NF40\_225155). All authors are with the Automatic Control Laboratory (IfA), ETH Z\"urich, 8092 Z\"urich, Switzerland \texttt{\small$\{$rzuliani,jlygeros$\}$@ethz.ch}. E. C. Balta is also with inspire AG, 8005 Z\"urich, Switzerland. \texttt{\small efe.balta@inspire.ch}.}}
\begin{document}

\maketitle
\begin{abstract}
Model mismatch and process noise are two frequently occurring phenomena that can drastically affect the performance of model predictive control (MPC) in practical applications.
We propose a principled way to tune the cost function and the constraints of linear MPC schemes \rz{to improve the closed-loop performance} and robust constraint satisfaction on uncertain nonlinear dynamics with additive noise.
The tuning is performed using a novel MPC tuning algorithm based on backpropagation developed in our earlier work.
Using the scenario approach, we provide probabilistic bounds on the likelihood of closed-loop constraint violation over a finite horizon.
We showcase the effectiveness of the proposed method on linear and nonlinear simulation examples.
\end{abstract}
\section{Introduction}
Model predictive control (MPC) is a model-based control technique that can efficiently solve challenging control tasks under input and process constraints by formulating, at each time step, a receding horizon optimization problem.
The mismatch between the nominal model used by the MPC and the true dynamics poses an important challenge in maintaining good closed-loop performance and ensuring constraint satisfaction. 
Many robust MPC methods have been developed in the literature, often relying on constraint tightening or probabilistic satisfaction guarantees.
\rz{However, tightenings are generally designed without explicitly accounting for the receding horizon aspect, and this can lead to overly conservative results.}
Here, we study the MPC problem with model uncertainty and provide a structured way to \rz{design constraint tightenings tailored to closed-loop operation, thus reducing conservatism.}

Tube MPC is a principled way to robustify MPC schemes whenever the process dynamics are unknown or subject to disturbances \cite{mayne2011tube}.
This strategy tightens the MPC constraints so that the resulting closed-loop state-input trajectory satisfies the nominal constraints. 
The tightening is generally designed based on the support of the uncertainty/noise set, which is assumed to be bounded \cite{chisci2001systems}.
Tube MPC schemes have been developed for linear systems subject to bounded additive noise \cite{chisci2001systems}, multiplicative uncertainty \cite{fleming2014robust}, and parametric uncertainty \cite{bujarbaruah2019adaptive}.
Moreover, extensions to nonlinear Tube MPC to deal with additive noise \cite{mayne2011tube} and model uncertainty \cite{kohler2020computationally} have been developed.
Despite its popularity, Tube MPC can be conservative, since constraint tightenings are often designed for the worst-case uncertainty realization, which is unlikely to occur in many practical applications, leading to cautious MPC designs.
Moreover, nonlinear tube-based solutions can be cumbersome to implement numerically and may require significant tuning effort \cite{ross2018scaling}.

A way to reduce conservatism is to construct a representation of the uncertain elements (either implicitly or explicitly) using data and derive probabilistic bounds on the likelihood of constraint satisfaction.
A notable example is the \textit{scenario approach} \cite{calafiore2006scenario}, where samples of the uncertain parameters (called scenarios) are used to obtain a control scheme with good out-of-sample performance.
Unlike tube MPC, the scenario approach can be applied without accurate knowledge of the underlying uncertainty distribution \rz{or support}.
However, constraint satisfaction is guaranteed only in probability instead of in the worst case, where a smaller constraint violation probability likely produces a more conservative performance.
For example, \cite{calafiore2012robust} proposes a scenario approach-based MPC design for uncertain linear systems subject to additive disturbances and derives guarantees on the closed-loop probability of constraint violation at each time step. 
The scheme of \cite{schildbach2014scenario}, in a similar setting, is guaranteed to have a small average constraint violation. 
The scenario approach can also be used in settings where the model dynamics are completely unknown \cite{micheli2022scenario}.
Existing methods, however, are almost exclusively limited to linear system dynamics, or have guarantees for single time-steps, providing little insight into the behavior over closed-loop trajectories.

In this paper, we design the cost and the constraints of an MPC scheme to maximize \rz{closed-loop} performance while ensuring robust constraint satisfaction. 
Our contribution is twofold: i) we provide a novel approach for optimal \rz{closed-loop} tuning of robust nonlinear MPC problems, \rzz{extending \cite{zuliani2023bp} to uncertain and noisy dynamics}, and ii) we use the scenario approach to provide sample-efficient guarantees on the closed-loop probability of constraint violation.
The tuned MPC can be formulated as a convex quadratic program \rz{even for nonlinear dynamics}, and can hence be solved efficiently and reliably with specialized software.
The design parameters are the terminal cost and the input cost of the MPC, as well as linear constraint tightenings.
All variables are tuned using the recently proposed \textit{BackPropagation-MPC} (\textit{BP-MPC}) algorithm \cite{zuliani2023bp}, which can achieve optimal closed-loop MPC designs using a sensitivity-based procedure.
Since the sensitivity information involves the closed-loop trajectory, our method greatly reduces the conservatism compared to existing offline-designed tube-based techniques. 

\paragraph*{Related Work}
\rzz{Tuning the cost and the constraints of MPC controllers is a concept that first appeared in \cite{amos2018differentiable}. In a stream of works, Zanon and Gros studied the interaction between nonlinear MPC and Reinforcement Learning (see \cite{gros2019data} and subsequent publications) and proposed a policy gradient method to improve the closed-loop performance of MPC. All these works, however, lack rigorous convergence guarantees and implicitly assume the MPC controller to be a continuously differentiable function of its parameters. Even under appropriate constraint qualifications, continuous differentiability can never hold everywhere (see the discussion in \cref{subsection:cons_jac} or \cite{bolte2024differentiating} for a comprehensive overview on the notion of differentiability for solution maps of optimization problems). In \cite{zanon2020safe}, the authors propose a tuning mechanism involving a tube MPC scheme to ensure robustness against model uncertainty and stochastic noise. However, this approach relies on traditional tube MPC architectures and, as a result, may suffer from conservativeness.}

\paragraph*{Outline} \rzz{The remainder of the paper is structured as follows. In \cref{section:PF}, we introduce the control problem and the MPC policy. \cref{section:NP} describes how the nominal cost of the MPC can be tuned to enhance nominal performance. In \cref{section:RCS}, we optimize the constraint tightenings to ensure robust constraint satisfaction in closed-loop operation and evaluate the scheme’s out-of-sample performance. Finally, in \cref{section:SIM}, we demonstrate the effectiveness of our approach through simulations.}

\paragraph*{Notation}
$\Z_{[a,b]}$ denotes integers between $a$ and $b$.
$x \sim \P$ means that $x$ is drawn from the probability distribution $\P$.
$\mathbb{E}[x]$ and $\mathbb{P}[x]$ denote expectation and probability of the random variable $x$, respectively. \rzz{$A \succ 0$ means that $A$ is symmetric positive definite.}

\section{Problem formulation}
\label{section:PF}
We consider an uncertain nonlinear system subject to additive disturbances
\begin{align}\label{eq:PF_sys}
x_{t+1}=f(x_t,u_t,d)+w_t,~ x_0 \sim \P_{x_0},
\end{align}
where $x_t\in\R^{n_x}$ and $u_t\in\R^{n_u}$ denote the state and input at time $t$, respectively, and $\P_{x_0}$ is an unknown distribution with known mean $\bar{x}_0$.
The parameter $d\in\R^{n_d}$ is a random variable representing model uncertainty with unknown distribution $\P_d$.
The additive noise $w_t\in\R^{n_w}$ is drawn i.i.d. for every $t$ from the unknown distribution $\P_w$.
The system needs to satisfy the following state and input constraints for all $t$
\begin{align}\label{eq:PF_cst}
H_x x_t \leq h_x ,~~~ H_u u_t \leq h_u.
\end{align}
We consider the case where the input $u_t$ is determined online by an MPC policy $u_t=\mpc(x_t,p,\eta)$, where $p$ and $\eta$ are design parameters to be defined shortly.
The closed-loop dynamics are then given by
\begin{align}\label{eq:PF_closed_loop}
x_{t+1} = f(x_t,\mpc(x_t,p,\eta),d)+w_t.
\end{align}
The nominal dynamics can be obtained from \cref{eq:PF_closed_loop} by setting $w_t=0$ and $d$ to some nominal value that (without loss of generality) we assume to be $0$, leading to
\begin{align}\label{eq:PF_sys_nominal}
\bar{x}_{t+1}=f(\bar{x}_t,\mpc(\bar{x}_t,p,\eta))
\end{align}
where, with a slight abuse of notation, we set $f(\bar{x}_t,\bar{u}_t):=f(\bar{x}_t,\bar{u}_t,0)$, and $\bar{x}_t$ and $\bar{u}_t = \mpc(\bar{x}_t,p,\eta)$ denote the nominal state and the nominal input, respectively.

Our goal is to design an MPC policy that steers the system to the origin while satisfying \cref{eq:PF_cst} for all possible $w_t\sim\P_w$, $d\sim \P_d$, and $x_0\sim\P_{x_0}$ within a finite time horizon $T>0$.
Both these objectives are captured by the following optimization problem.
\begin{align}\label{eq:PF_problem}
\begin{split}
\underset{p,\eta,x,u}{\mathrm{minimize}} & \quad \mathbb{E}_{w,d,x_0} \left[ \sum_{t=0}^{T} \|x_t\|_{Q_x}^2 \right]\\
\mathrm{subject~to} & \quad x_{t+1} = f(x_t,u_t,d)+w_t, \\
& \quad u_t = \mpc(x_t,p,\eta),\\
& \quad H_x x_t \leq h_x, ~ H_u u_t \leq h_u,\\
& \quad \forall  w_t,\, d, \, x_0,~ \forall t \in \Z_{[0,T]}.
\end{split}
\end{align}
where $Q_x \succ 0$ and $w:=(w_0,\dots,w_T)$. While our framework readily allows more complex cost functions, we limit the cost to be a quadratic function of $x$ for simplicity, and refer the reader the reader to \cite[Section VI-C]{zuliani2023bp} for the general case.

We focus on MPC policies that can be expressed as strongly convex quadratic programs.
Specifically, given two design parameters $p:=(P,R)$, with $P,R \succ 0$ (terminal and input cost), and $\eta:=(\eta_x,\eta_u)$ (state and input constraint tightenings), we choose $\mpc(x_t,p,\eta)=v_{0|t}$ by solving
\begin{align}\label{eq:PF_mpc}
\begin{split}
\underset{z_t,v_t}{\mathrm{min.}} & \quad \|z_{N|t}\|_P^2 + \sum_{k=0}^{N-1} \|z_{k|t}\|_{Q_x}^2 + \|v_{k|t}\|_R^2\\
\mathrm{s.t.} & \quad z_{k+1|t}=A_{k|t}z_{k|t}+B_{k|t}v_{k|t}+c_{k|t}, ~ z_{0|t} = x_t, \\
& \quad H_x z_{k|t} \leq h_x - \eta_{x,k}^2 ,~ H_u v_{k|t} \leq h_u - \eta_{u,k}^2,\\
& \quad \forall k \in \Z_{[0,N-1]},
\end{split}
\end{align}
where $z_t:=(z_{0|t},\dots,z_{N|t})$, $v_t:=(v_{0|t},\dots,v_{N-1|t})$, $\eta_x:=(\eta_{x,0},\dots,\eta_{x,N})$, $\eta_u:=(\eta_{u,0},\dots,\eta_{u,N-1})$, and the square in the tightenings is applied elementwise.
The prediction horizon $N$ of the MPC is generally much smaller than $T$.
Since \cref{eq:PF_mpc} may become infeasible in practice, we relax the state constraints with the technique of \cite[Section VI-D]{zuliani2023bp}.

The equality constraints in \cref{eq:PF_mpc} should be designed to ensure that $A_{k|t}z_{k|t}+B_{k|t}v_{k|t}+c_{k|t}\approx f(z_{k|t},v_{k|t})$ for all $k\in\Z_{[0,N-1]}$. To this end, denoting with $(z_{t-1},v_{t-1})$ the optimal state-input trajectory obtained by solving \cref{eq:PF_mpc} at time-step $t-1$, we set
\begin{align*}
&A_{k|t} = \frac{\partial f}{\partial x}(z_{k+1|t-1},v_{k+1|t-1}),\\
&B_{k|t} = \frac{\partial f}{ \partial u}(z_{k+1|t-1},v_{k+1|t-1}),\\
&c_{k|t} = f(z_{k+1|t-1},v_{k+1|t-1}) - A_{k|t}z_{k+1|t-1} - B_{k|t}v_{k+1|t-1}.
\end{align*}
For simplicity, we assume that $A_{k|t}\equiv A$, $B_{k|t}\equiv B$ and $c_{k|t}\equiv 0$ and refer the reader to \cite[Section VI-A]{zuliani2023bp}. \rz{Observe that linear dynamics are only used within the MPC problem \cref{eq:PF_mpc}, whereas the true nonlinear dynamics \cref{eq:PF_sys} are used in \cref{eq:PF_problem}}.

\section{Improving nominal performance}
\label{section:NP}
To solve \cref{eq:PF_problem}, we first design $\theta:=(p,\eta)$ to minimize the cost in \cref{eq:PF_problem} for the nominal dynamics \cref{eq:PF_sys_nominal} by solving
\begin{align}\label{eq:NP_nom_perf}
\begin{split}
\underset{\theta,\bar{x}}{\mathrm{minimize}} & \quad \sum_{t=0}^{T} \|\bar{x}_t\|_{Q_x}^2 \\
\mathrm{subject~to} & \quad \bar{x}_{t+1} = f(\bar{x},\mpc(\bar{x},\theta)), ~ \bar{x}_0 \text{ given},\\
& \quad H_x \bar{x}_t \leq h_x, ~ \forall t \in \Z_{[0,T]}.
\end{split}
\end{align}
We omit the input constraints since the MPC policy \cref{eq:PF_mpc} satisfies them by design.
For simplicity, let $x(\theta):=(\bar{x}_0(\theta),\dots,\bar{x}_T(\theta))$ be the function mapping $\theta$ to the nominal closed-loop trajectory $\bar{x}(\theta)$ obtained by setting $\bar{x}_0(\theta)=\bar{x}_0$ and by iterating the nominal dynamics $\bar{x}_{t+1}=f(\bar{x}_t,\mpc(\bar{x}_t,\theta))$ until time-step $T$.
\rz{Using an exact penalty function, we can reformulate \cref{eq:NP_nom_perf} as the unconstrained minimization problem
\begin{align}\label{eq:NP_nom_perf_unc}
\underset{\theta}{\mathrm{minimize}} ~~ \ell(\bar{x}(\theta)) := \sum_{t=0}^{T} \|x_t(\theta)\|_{Q_x}^2 + c_1 \gamma(x_t(\theta)) ,
\end{align}
where $\gamma(x) := \|\max\{ H_x x - h_x, 0 \} \|_1$, $c_1>0$, and $\max$ is applied elementwise.}
If \cref{eq:NP_nom_perf} is sufficiently well-behaved, and $c_1$ is large enough, \cref{eq:NP_nom_perf} and \cref{eq:NP_nom_perf_unc} are equivalent.
\begin{definition}
\rzz{Problem \cref{eq:NP_nom_perf} is calm at $\theta^*$ if $H_x x_t(\theta^*) \leq h_x$ for all $t$, and there exists $\bar{\alpha}\geq 0$ such that for all $\theta$ sufficiently close to $\theta^*$
\begin{align}
\sum_{t=0}^T \|x_t(\theta)\|_{Q_x}^2 + \bar{\alpha} \gamma(x_t(\theta)) \geq \sum_{t=0}^T \|x_t(\theta^*)\|_{Q_x}^2. \label{eq:calmness_condition}
\end{align}
The constant $\bar{\alpha}$ is called \emph{calmness module}.}
\end{definition}

\smallskip

Calmness is a weak constraint qualification that is verified in many situations. For details, we refer the reader to \cite{burke1991exact}.
\begin{lemma}[{\cite[Theorem 2.1]{burke1991exact}}]\label{lemma:NP_exact_penalty}
The set of calm local minima of \cref{eq:NP_nom_perf} coincide with the set of local minima of \cref{eq:NP_nom_perf_unc} if $c_1$ is chosen at least as large as the calmness modulus.
\end{lemma}

\rzz{Computing the calmness modulus can be challenging for nonconvex problems.} \rz{In practice, we can expect $\theta^*$ to be a local minimizer of \cref{eq:NP_nom_perf} if $c_1$ is chosen large enough.} \rzz{We further investigate the impact of $c_1$ in \cref{subsection:penalty_example}.}

\subsection{Conservative Jacobians} \label{subsection:cons_jac}

Problem \cref{eq:NP_nom_perf_unc} can be solved using a simple gradient-based scheme. However, since the cost function in \cref{eq:NP_nom_perf_unc} is typically nondifferentiable \rz{due to the nonsmoothness introduced by the MPC}, a more general notion of gradient is needed. To this end, we use the concept of conservative Jacobians \cite{bolte2021nonsmooth}.
\begin{definition}[{\cite[Section 2]{bolte2021nonsmooth}}]

Let $\varphi:\R^n\to\R^m$ be a locally Lipschitz function. We say that the set-valued function $\J_\varphi:\R^n\tto\R^m$ is a \textit{conservative Jacobian} for $\varphi$, if $\J_\varphi$ is nonempty-valued, outer semicontinuous, locally bounded, and for all paths\footnote{A \textit{path} is an absolutely continuous function $\rho:[0,1]\to\R^n$ admitting a derivative $\dot{\rho}$ for almost every $t\in[0,1]$ and for which the Lebesgue integral of $\dot{\rho}$ between $0$ and any $t\in[0,1]$ equals $\rho(t)-\rho(0)$.} $\rho:[0,1]\to\R^n$ and almost all $t\in[0,1]$ it holds that $\frac{\mathrm{d} \varphi}{\mathrm{d} t}(\rho(t)) = \langle v, \dot{\rho}(t) \rangle ,~~~ \forall v \in \J_\varphi(\rho(t))$. A function is \textit{path-differentiable} if it admits a conservative Jacobian.
\end{definition}
\begin{remark}
\rzz{Conservative Jacobians are particularly important in the context of MPC, where traditional Jacobians, as well as more general notions such as Clarke Jacobians, do not suffice. Typically, the "derivative" of the solution map of an optimization problem is computed by applying the implicit function theorem (IFT) to the optimality conditions of the problem \cite{amos2018differentiable,gros2019data,zanon2020safe}. However, the classical (i.e., smooth) IFT assumes that the underlying function is continuously differentiable, an assumption that does not hold in the MPC setting, where the solution map can fail to be differentiable due to changes in the set of active inequality constraints. As a result, traditional Jacobians cannot be used to analyze the differentiability properties of MPC. Likewise, Clarke Jacobians are inadequate in this context, as they do not satisfy an implicit function theorem (see \cite[Example 1]{bolte2021nonsmooth}).}
\end{remark}

Given two path-differentiable functions $\varphi:\R^n\to\R^m$ and $\chi:\R^m\to\R^p$, the function $\psi:=\chi\circ\varphi$ is path-differentiable with $\J_\psi(z)=\J_\varphi(\chi(z))\J_\chi(z)$.
Importantly, not all locally Lipschitz functions are path-differentiable.
In this paper, we focus on the class of definable functions.

\begin{definition}[{\cite[Definitions 1.4 and 1.5]{coste1999introduction}}]
\rz{A collection $\mathcal{O}=(\mathcal{O}_n)_{n\in\N}$, where each $\mathcal{O}_n$ contains subsets of $\R^n$, is an \textit{o-minimal structure} on $(\R,+,\cdot)$ if
\begin{enumerate}
    \item all semialgebraic subsets of $\R^n$ belong to $\mathcal{O}_n$;
    \item the elements of $\mathcal{O}_1$ are precisely the finite unions of points and intervals;
    \item $\mathcal{O}_n$ is a boolean subalgebra of the powerset of $\R^n$;
    \item if $A\in \mathcal{O}_n$ and $B\in\mathcal{O}_m$, then $A\times B \in \mathcal{O}_{n+m}$;
    \item if $A\in\mathcal{O}_{n+1}$, then the set containing the elements of $A$ projected onto their first $n$ coordinates belongs to $\mathcal{O}_n$.
\end{enumerate}
A subset of $\R^n$ which belongs to $\mathcal{O}$ is said to be \textit{definable} (in the o-minimal structure). A function $\varphi:\R^n\to\R^p$ is \textit{definable} if its graph $\{(x,v):v=\varphi(x)\}$ is definable.}
\end{definition}

Locally Lipschitz definable functions are ubiquitous in control and optimization, and admit a conservative Jacobian.
Moreover, they can be minimized (locally) with \cref{alg:NP_min_of_path_diff}, which is guaranteed to converge to a critical point for a suitable choice of step sizes.

\begin{algorithm}
\caption{Minimization of path-differentiable function}\label{alg:NP_min_of_path_diff}
\begin{algorithmic}[1]
\Input $x^0$, $\{\alpha_k\}_{k\in\N}$, $\texttt{tol}>0$.
\For{$k=1$ to $\infty$}
\State \rz{Compute any} $p^k\in \mathcal{J}_\varphi(x^k)$
\State $x^{k+1}=x^k-\alpha_kp^k$
\State \textbf{If} $\|x^{k}-x^{k-1}\|_2 < \texttt{tol}$ \Return $x^*=x^{k+1}$
\EndFor
\end{algorithmic}
\end{algorithm}

\begin{lemma}[{\cite[Theorem 6.2]{davis2020stochastic}}]\label{lemma:NP_sgd_converges}
Let $\varphi$ be locally Lipschitz and definable in some o-minimal structure, assume that $\sup_k \|x^k\|_2 < \infty$ and that
\begin{align}\label{eq:NP_stepsizes}
\sum_{k=0}^{\infty} \alpha_k = \infty ,~~ \sum_{k=0}^{\infty} \alpha_k^2 < \infty.
\end{align}
\rz{Then $\{x^k\}_{k\in\N}$ obtained with \cref{alg:NP_min_of_path_diff} converges to some $x^*$ satisfying $ 0\in \J_\varphi(x^*)$.}
\end{lemma}

One way to guarantee bounded iterates $x^k$ is to introduce a projection to a large enough polytopic set $\mathcal{X}$ in the gradient descent update, i.e., $x^{k+1}=\mathcal{P}_{\mathcal{X}}[x^k-\alpha_kp^k]$ (see discussion in \cite[Section 6.1]{davis2020stochastic}). \rz{To guarantee \cref{eq:NP_stepsizes} one can choose
\begin{align}\label{eq:NP_stepsize_rule}
\alpha_k = \frac{c}{k^\zeta} ~~ c>0,~~ \zeta \in (0.5,1].
\end{align}}
In this paper we consider a fixed o-minimal structure $\mathcal{O}$ and assume that all definable functions are definable in $\mathcal{O}$.

To ensure that \cref{eq:NP_nom_perf_unc} can be solved with \cref{alg:NP_min_of_path_diff}, we need $\ell(\bar{x}(\theta))$ to be locally Lipschitz and definable. This is the case if $\ell$ and $\bar{x}$ are locally Lipschitz definable, as both these properties are preserved by composition.
\begin{assumption}\label{ass:NP_cost}
The cost $\ell$ is locally Lipschitz definable.
\end{assumption}

\subsection{The BP-MPC algorithm}

The BP-MPC algorithm \cite{zuliani2023bp} uses backpropagation to efficiently construct $\J_{x}$ for a given $d$ recursively 
\begin{align}\label{eq:NP_backprop}
\J_{x_{t+1}}(\theta) = & \, \J_{f,u}(x_t,u_t,d)\left[ \J_{\mpc,x_t}(x_t,\theta)\J_{x_t}(p) \right. \notag \\  & \left. + \, \J_{\mpc,\theta}(x_t,\theta) \right] + \J_{f,x}(x_{t},u_{t},d)\J_{x_t}(p),
\end{align}
where $\J_{f,x}$ and $\J_{f,u}$ are the partial conservative Jacobians of $f$ with respect to its arguments (and similarly for $\J_{\mpc,x_t}$ and $\J_{\mpc,\theta}$), and $\J_{x}(\theta)=0$, since $x_0$ is independent of $\theta$.
We provide here a general algorithm that works for any value of $w$, $d$, and $x_0$, and later consider the nominal case.
To apply \cref{eq:NP_backprop} we require the following.

\begin{assumption}\label{ass:NP_f_local_lip_semialg}
The function $f$ is locally Lipschitz and definable in $(x,u)$ for all $d\sim\P_d$.
\end{assumption}

To compute the conservative Jacobian $\J_{\mpc}$ of the MPC map, we rewrite \cref{eq:PF_mpc} as a quadratic program in standard form
\begin{align}\label{eq:NP_primal}
\begin{split}
\underset{y}{\mathrm{minimize}} & \quad \frac{1}{2}y^\top Q(p) y + q(x_t,p)^\top y \\
\mathrm{subject~to} & \quad G(p)y \leq g(x_t,p,\eta),\\
& \quad F(p)y = \phi(x_t,p),
\end{split}
\end{align}
and obtain its Lagrange dual
\begin{align}\label{eq:NP_dual}
\begin{split}
\underset{z}{\mathrm{minimize}} & \quad \frac{1}{2}z ^\top H(p) z + h
(x_t,p,\eta)^\top z\\
\mathrm{subject~to} & \quad z = (\lambda,\mu) \in \R^{n_\text{in}} \times \R^{n_\text{eq}},~ \lambda \geq 0.
\end{split}
\end{align}
Note that both problems \cref{eq:NP_primal} and \cref{eq:NP_dual} do not depend on the choice of $d$ and $w$, since the MPC \cref{eq:PF_mpc} utilizes an approximation of the nominal dynamics \cref{eq:PF_sys_nominal}.
The solution $y(\bar{p})$ of \cref{eq:NP_primal}, where $\bar{p}:=(x_t,p,\eta)$, is obtained from the solution $z(\bar{p})$ of \cref{eq:NP_dual} as
$y(\bar{p})=\mathcal{G}(z(\bar{p}),\bar{p})$, where
\begin{align*} 
\mathcal{G}(z(\bar{p}),\bar{p}) := -Q(p)^{-1} \left( [G(p)~F(p)]^\top z(\bar{p}) + q(\bar{p}) \right).
\end{align*}
The existence of $\J_y$ can be guaranteed under the following assumptions.

\begin{assumption}\label{ass:NP_local_lip_semialg}
The maps $Q(p)$, $q(\bar{p})$, $G(p)$, $g(\bar{p})$, $F(p)$, and $\phi(\bar{p})$ are locally Lipschitz definable. Moreover, $Q^{-1}(p)$ is locally Lipschitz.
\end{assumption}

\begin{assumption}\label{ass:NP_feas_strong_conv_licq}
For all values of $x_t$, $p$, and $\eta$, problem \cref{eq:NP_primal} is feasible, strongly convex, and satisfies the linear independence constraint qualification (LICQ).
\end{assumption}

Assumption \ref{ass:NP_local_lip_semialg} is not restrictive in practice, as the class of locally Lipschitz definable functions comprises most functions commonly used in control and optimization (e.g., semialgebraic, trigonometric restricted to a compact definable domain, exponential function). \rzz{Moreover, any combination of definable functions (such as addition, multiplication, power, differentiation, composition) remains definable. A rich body of literature exists on definable functions; for a comprehensive overview, we refer the reader to \cite{coste1999introduction}.}

\rzz{Within \cref{ass:NP_feas_strong_conv_licq}, the feasibility requirement is not restrictive as state constraints can be relaxed by introducing slack variables. To ensure that the controller favors solutions that fulfill the constraints whenever possible, a penalty on the slack variables must be included in the MPC cost. This procedure is described in detail in \cite[Section VI-D]{zuliani2023bp}.} The LICQ assumption holds e.g. if the constraints in \cref{eq:PF_cst} are box constraints $x_\text{min} \leq x_t \leq x_\text{max}$, $u_\text{min} \leq u_t \leq u_\text{max}$. \rz{The convexity assumption can be ensured by design with a suitable parameterization of $Q$. Note that problem \cref{eq:PF_problem} remains nonconvex despite Assumption \ref{ass:NP_feas_strong_conv_licq}.}
\begin{proposition}[{\cite[Theorem 1]{zuliani2023bp}}]\label{prop:NP_bpmpc}
Under Assumptions \ref{ass:NP_local_lip_semialg} and \ref{ass:NP_feas_strong_conv_licq}, the optimizer $z(\bar{p})$ of \cref{eq:NP_dual} is unique and locally Lipschitz definable. Its conservative Jacobian $\J_z(\bar{p})$ contains elements of the form $-U^{-1}V$, where
\begin{align*}
U &\in T (I - \gamma H(p))-I,\\
V &\in -\gamma T(Az+B),\\
T &= \operatorname{diag}(\operatorname{sign}(\lambda_1),\dots,\operatorname{sign}(\lambda_{n_{\text{in}}}),1,\dots,1),
\end{align*}
where $z=(\lambda,\mu)$, $A\in\J_{H}(p)$, $B\in\J_{h}(\bar{p})$, and $\gamma$ is any positive constant.
Moreover, the optimizer $y(\bar{p})$ of \cref{eq:NP_primal} is unique and locally Lipschitz definable with conservative Jacobian
\begin{align*}
W-Q(p)^\top[G(p)^\top~F(p)^\top]Z \in \J_{y}(\bar{p}),
\end{align*}
where $Z\in\J_z(\bar{p})$ and $W \in \mathcal{J}_{\mathcal{G},\bar{p}}(z(\bar{p}),\bar{p})$.
\end{proposition}

Proposition \ref{prop:NP_bpmpc} provides a way to compute the conservative Jacobian $\J_{\mpc}$ of the MPC map.
Combining with \cref{eq:NP_backprop}, we can iteratively construct the conservative Jacobian $\J_{x}$ of the closed-loop trajectory $x$ for any value of $w$, $d$, $x_0$.
The procedure is summarized in \cref{alg:NP_cons_jac}.

\begin{algorithm}
\caption{Conservative Jacobian computation}\label{alg:NP_cons_jac}
\begin{algorithmic}[1]
\Input $\theta$, $w$, $d$, $x_0$.
\Init $\J_{x_0}(\theta)=0$.
\For{$t=0$ to $T$}
    \State Solve \cref{eq:PF_mpc} and set $u_t = \mpc(x_t,\theta)$.
    \State Compute $\J_{x_{t+1}}(\theta)$ using \cref{eq:NP_backprop} and \cref{prop:NP_bpmpc}.
    \State Compute next state $x_{t+1}=f(x_t,u_t,d) + w_t$.
\EndFor
\State\Return $\J_{x}(\theta)$
\end{algorithmic}
\end{algorithm}

To compute the conservative Jacobian of $\bar{x}$ for a given $\theta$ it suffices to set $w=0$, $d=0$, and $x_0=\bar{x}_0$ in \cref{alg:NP_cons_jac}.

\subsection{A gradient-based solution}

Once the conservative Jacobian of the closed-loop trajectory $\bar{x}$ is available, we can obtain the conservative Jacobian of the objective in \cref{eq:NP_nom_perf_unc} using the chain rule $\J_{\ell}(\theta)=\J_\ell(\bar{x})\J_{\bar{x}}(\theta)$.
Combining this with \cref{alg:NP_min_of_path_diff,alg:NP_cons_jac}, we obtain \cref{alg:NP_bpmpc_nominal}, which converges to a critical point of \cref{eq:NP_nom_perf_unc}.

\begin{algorithm}
\caption{BP-MPC for Nominal Performance}\label{alg:NP_bpmpc_nominal}
\begin{algorithmic}[1]
\Input $\theta^0$, $\{ \alpha_k \}_{k\in\N}$, $\texttt{tol}>0$.
\For{$k=0$ to $\infty$}
    \State Compute $J_1^{k} \in \J_{\bar{x}}(\theta^k)$ with \cref{alg:NP_cons_jac}.
    \State Compute \rz{any} $J^{k}_2\in\J_{\ell}(\bar{x})$.
    \State Compute $J^k = J_2^{k} J_{1}^{k}$.
    \State Update $\theta^{k+1}=\theta^k-\alpha_kJ^k$.
    \State \textbf{If} $\|\theta^{k+1}-\theta^{k}\|_2<\texttt{tol}$ \Return $\theta^*=\theta^{k+1}$
\EndFor
\end{algorithmic}
\end{algorithm}

\begin{theorem}\label{thm:NP_bpmpc_nominal}
Suppose that Assumptions \ref{ass:NP_cost}, \ref{ass:NP_f_local_lip_semialg}, \ref{ass:NP_local_lip_semialg}, and \ref{ass:NP_feas_strong_conv_licq} hold, that $\{ \alpha_k \}_{k\in\N}$ is chosen as in \cref{lemma:NP_sgd_converges}, and that $\sup_k\|p^k\|_2< \infty$. {Then $\{\theta^k\}_{k\in\N}$ obtained with \cref{alg:NP_bpmpc_nominal} converges to a critical point $\theta^*$ of \cref{eq:NP_nom_perf_unc}}. Moreover, if \cref{eq:NP_nom_perf} is calm at $\theta^*$, and $c_1$ in \cref{eq:NP_nom_perf_unc} is chosen at least as large as the calmness modulus, then $\theta^*$ is also a local minimizer of \cref{eq:NP_nom_perf}.
\end{theorem}
\begin{proof}
The first part follows immediately by recognizing that \cref{alg:NP_bpmpc_nominal} is implementing a gradient-descent rule equivalent to that in \cref{alg:NP_min_of_path_diff}, and by applying Lemma \ref{lemma:NP_sgd_converges}. The second follows from Lemma \ref{lemma:NP_exact_penalty}.
\end{proof}

\section{Robust constraint satisfaction}
\label{section:RCS}
We now focus on ensuring robust constraint satisfaction by solving the problem
\begin{align}\label{eq:RCS_problem}
\begin{split}
\underset{\tilde{\theta},x,u}{\mathrm{minimize}} & \quad \|\tilde{\theta}-\theta^*\|_2^2\\
\mathrm{subject~to} & \quad x_{t+1}=f(x_t,u_t,d)+w_t,\\
& \quad u_t = \mpc(x_t,\tilde{\theta}),\\
& \quad H_x x_t \leq h_x,~H_uu_t \leq h_u,\\
& \quad \forall  w_t,\, d, \, x_0, ~ \forall t \in\Z_{[0,T]},
\end{split}
\end{align}
where $\theta^*$ is the solution of \cref{eq:NP_nom_perf} obtained with \cref{alg:NP_bpmpc_nominal}.
By penalizing the difference between $\tilde{\theta}$ and $\theta^*$, we ensure that $\tilde{\theta}^* \approx \theta^*$ while satisfying the constraints.

\subsection{Robust constraint satisfaction using BP-MPC}
We assume that a set of i.i.d. samples is available:
\begin{align*}
\mathcal{S} := \{ (w^j,d^j,x_0^j)_{j=1}^M : w^j \sim \P_w^T,~ d^j \sim \P_d,~ x_0^j \sim \P_{x_0} \}.
\end{align*}
\rz{Denoting $\delta=(w,d,x_0)\in\mathcal{S}$, Problem \cref{eq:RCS_problem} becomes}
\begin{align}\label{eq:RCS_problem_tractable}
\begin{split}
\underset{\tilde{\theta},x^\delta,u^\delta}{\mathrm{minimize}} & \quad \|\tilde{\theta}-\theta^*\|_2^2\\
\mathrm{subject~to}~~\, & \quad x_{t+1}^\delta=f(x_t^\delta,u_t^\delta,d)+w_t,\\
& \quad u_t^\delta = \mpc(x_t^\delta,\tilde{\theta}),\\
& \quad H_x x_t^\delta \leq h_x, ~ \forall \delta = (w,d,x_0) \in \mathcal{S}.
\end{split}
\end{align}
With the same strategy as in \cref{section:NP}, we can remove all constraints from \cref{eq:RCS_problem_tractable} using a penalty function 
\begin{align}\label{eq:RCS_problem_tractable_unc}
\underset{\tilde{\theta}}{\mathrm{min.}} ~ \, \|\tilde{\theta}-\theta^*\|_2^2  +  c_1 \sum_{\delta\in\mathcal{S}} \sum_{t=0}^{T} \gamma(x_t^\delta(\tilde{\theta})),
\end{align}
where $x^\delta(\tilde{\theta}):=(x^\delta_0(\tilde{\theta}),\dots,x^\delta_T(\tilde{\theta}))$ is the function mapping $\tilde{\theta}$ to the closed-loop trajectory $x^\delta(\tilde{\theta})$ obtained by setting $x_0^\delta(\tilde{\theta})=x_0$ and by iterating \cref{eq:PF_sys} until time-step $T$ with parameters $\tilde{\theta}$, $w$, and $d$, with $\delta=(w,d,x_0)$.
%

To facilitate the task of quantifying the robustness of the solution, we solve \cref{eq:RCS_problem_tractable_unc} using the Pick2Learn (P2L) algorithm \cite{paccagnan2024pick}, outlined in \cref{alg:RCS_bpmpc_constraints}. P2L converges to a local solution of \cref{eq:RCS_problem_tractable} under appropriate calmness assumptions.

\begin{algorithm}
\caption{BP-MPC for Robust Constraint Satisfaction}
\label{alg:RCS_bpmpc_constraints}
\begin{algorithmic}[1]
\Input $\theta^*$, $\mathcal{S}$,
\Init $\tilde{\theta}^0=\theta^*$, $\mathcal{T}^0=\emptyset$, $\mathcal{E}^0=\mathcal{S}$, $k=0$, $\texttt{converged}=\text{False}$.
\While{not $\texttt{converged}$}
    \State For all $\delta\in\mathcal{E}^k$, compute $\gamma^\delta := \sum_{t=0}^T\gamma(x_t^\delta(\tilde{\theta}^k))$.
    \If{$\gamma^\delta>0$ for some $\delta\in\mathcal{E}^k$}
        \State Select $\bar{\delta}=\operatorname{argmax}_\delta \gamma^\delta$.
    \ElsIf{$\sum_{t=0}^T\J_\gamma(x_t^\delta(\tilde{\theta}^k)) \neq \{ 0 \}$ for some $\delta\in\mathcal{E}^k$}
        \State Select any $\bar{\delta}$ with $\sum_{t=0}^T\J_\gamma(x_t^{\bar{\delta}}(\tilde{\theta}^k) \neq \{ 0 \}$.
    \Else
        \State $\texttt{converged}\gets \text{True}$
    \EndIf
    \State Update $\mathcal{T}^{k+1} = \mathcal{T}^k \cup \{ \bar{\delta} \}$, $\mathcal{E}^{k+1} = \mathcal{E}^k \setminus \{\bar{\delta}\}$.
    \State Solve \cref{eq:RCS_problem_tractable_unc} with $\mathcal{T}^{k+1}$ instead of $\mathcal{S}$ and obtain $\tilde{\theta}^{k+1}$.
    \State $k \gets k+1$
\EndWhile
\State \Return $\tilde{\theta}^*=\tilde{\theta}^{k}$ and $\mathcal{T}^*=\mathcal{T}^k$.
\end{algorithmic}
\end{algorithm}

P2L requires solving \cref{eq:RCS_problem_tractable_unc} several times for (much) smaller datasets $\mathcal{T}^k$ replacing $\mathcal{S}$. To do so, we use the scheme in \cref{alg:SGD}. Note that generally $|\mathcal{T}^*| \ll |\mathcal{S}|$ \cite{campi2018general}.

\begin{algorithm}
\caption{GD algorithm to solve \cref{eq:RCS_problem_tractable_unc}}
\label{alg:SGD}
\begin{algorithmic}[1]
\Input $\tilde{\theta}^k$, $\{ \alpha_j \}_{j\in\N}$, $\texttt{max\_it}\in\N$, $\mathcal{T}^k$,
\Init $\tilde{\theta}^{k,0}=\tilde{\theta}^k$.
\For{$j=0$ to $\texttt{max\_it}$}
    \For{$\delta\in\mathcal{T}^k$}
        \State Compute $J_1^{\delta}\in\J_{x^\delta}(\tilde{\theta}^{k,j})$ with \cref{alg:NP_cons_jac}.
        \State Compute any $J_2^{\delta} \in \sum_{t=0}^{T} \J_\gamma(x_t^\delta(\tilde{\theta}^{k,j}))$.
    \EndFor
    \State Compute gradient $J^{k,j}=2(\tilde{\theta}^{k,j}-\theta^*)+\sum_{\delta\in\mathcal{T}^k}J_2^\delta J_1^\delta$
    \State Update $\tilde{\theta}^{k,j+1}=\tilde{\theta}^{k,j}-\alpha_j J^{k,j}$.
\EndFor
\State \Return $\tilde{\theta}^{k+1}=\tilde{\theta}^{k,j+1}$.
\end{algorithmic}
\end{algorithm}
%
%

\begin{theorem}\label{thm:RCS_bpmpc_constraints}
Suppose that Assumptions \ref{ass:NP_f_local_lip_semialg}, \ref{ass:NP_local_lip_semialg}, and \ref{ass:NP_feas_strong_conv_licq} hold, that $\{ \alpha_j \}_{j\in\N}$ satisfy the stepsize condition in \cref{lemma:NP_sgd_converges}, and that in \cref{alg:SGD}, for any $\mathcal{T}^k$, $\sup_j\|\tilde{\theta}^{k,j}\|_2< \infty$. Then $\tilde{\theta}^k$ converges to a critical point $\tilde{\theta}^*$ of \cref{eq:RCS_problem_tractable_unc}. Moreover, if \cref{eq:RCS_problem_tractable} is calm at $\tilde{\theta}^*$, and $c_1$ in \cref{eq:RCS_problem_tractable_unc} is chosen at least as large as the calmness modulus, then $\tilde{\theta}^*$ is a local minimizer of \cref{eq:RCS_problem_tractable}.
\end{theorem}
\begin{proof}
By \cite[Theorem 3]{bolte2021nonsmooth}, we have that \cref{alg:SGD} converges to a critical point of \cref{eq:RCS_problem_tractable_unc} for any set of samples $\mathcal{S}$. Next, \cref{alg:RCS_bpmpc_constraints} must always converge in at most $|\mathcal{S}|$ iterations (as $\mathcal{E}^{|\mathcal{S}|}=\emptyset$). If \cref{alg:RCS_bpmpc_constraints} terminates after $k=|\mathcal{S}|$ iterations, then $\tilde{\theta}^*$ trivially solves \cref{eq:RCS_problem_tractable_unc}. If the algorithm terminates after $k<|\mathcal{S}|$ iterations, then $\tilde{\theta}^*$ satisfies
\begin{align*}
\tilde{\theta}^*-\theta^* + c_1 \sum_{\delta\in\mathcal{T}^*}\sum_{t=0}^T \J_\gamma(x_t^\delta(\tilde{\theta}^*)) \ni 0.
\end{align*}
Since for all $\delta\in \mathcal{S}\setminus \mathcal{T}^*$ we have $\J_\delta (\tilde{\theta}^*)=\{0\}$ for all $t$, $\tilde{\theta}^*$ also satisfies
\begin{align*}
\tilde{\theta}^*-\theta^* + c_1 \sum_{\delta\in\mathcal{S}}\sum_{t=0}^T \J_\gamma(x_t^\delta(\tilde{\theta}^*)) \ni 0,
\end{align*}
meaning that $\tilde{\theta}^*$ is a critical point of \cref{eq:RCS_problem_tractable_unc}. If the calmness assumption is met, then $\tilde{\theta}^*$ is a local minimizer of \cref{eq:RCS_problem_tractable} by Lemma \ref{lemma:NP_exact_penalty}.
\end{proof}

\subsection{Out-of-sample constraint satisfaction}\label{subsection:RCS_out_of_sample}
In this section, we study how well $\tilde{\theta}^*$ performs on unseen samples obtained from $\P_{w}$, $\P_{d}$, $\P_{x_0}$ (that is, assuming no distribution shift) by adapting the results of \cite{campi2018general}.
We want to ensure that the constraint violation probability
\begin{align*}
V(\tilde{\theta}^*):=\mathbb{P}_{w,d,x_0}\left\{ H_x x_t(\tilde{\theta}^*,w,d,x_0) > h_x ~ \forall t\in\Z_{[0,T]} \right\}.
\end{align*}
is smaller than a certain tolerance $\epsilon\in(0,1)$. Here $x(\tilde{\theta}^*,w,d,x_0)$ denotes the closed-loop trajectory obtained from \cref{eq:PF_sys} starting from $x_0$ with parameters $w$, $d$, and $\theta$, and the probability is with respect to $\P_{w}^{T} \times \P_d \times \P_{x_0}$.
Due to the probabilistic choice of $\mathcal{S}$, this statement is made with confidence $1-\beta$,
\begin{align}\label{eq:RCS_bound_scenario}
\mathbb{P}_\mathcal{S} \left\{ V(\tilde{\theta}^*) > \epsilon \right\} \leq \beta
\end{align}
with $\beta\in(0,1)$.
In \cref{eq:RCS_bound_scenario} the probability is with respect to the multi-sample $\mathcal{S}$ in \cref{eq:RCS_problem_tractable}, which is drawn from $(\P_{w}^{T} \times \P_d \times \P_{x_0})^M$.
If \cref{eq:RCS_bound_scenario} is satisfied for very small values of $\beta$, we can practically guarantee $V(\tilde{\theta}^*)\leq \epsilon$ \cite{campi2018general}.

The guarantees make use of the following notion.
\begin{definition}[Support subsample]\label{def:RCS_support_constraint}
Given a collection of samples $\mathcal{S}=\{ \delta^j, j\in\Z_{[1,M]} \}$, a \textit{support subsample} is a collection of $k$ elements $\mathcal{D}=\{ \delta^{j_i}: i\in \Z_{[1,k]} \}$, with $j_1<\dots<j_k$, such that solving \cref{eq:RCS_problem_tractable} with $\mathcal{S}$ replaced with $\mathcal{D}$ produces the same solution.
\end{definition}
%

\rz{Note that the set $\mathcal{T}^*$ returned by \cref{alg:RCS_bpmpc_constraints} is a support sub-sample of $\mathcal{S}$. The simplicity with which one can identify a support subsample, enabled by the P2L algorithm, is the primary reason why we decided to solve \cref{eq:RCS_problem} using \cref{alg:RCS_bpmpc_constraints} instead of \cref{alg:SGD} with $\mathcal{T}^k=\mathcal{S}$. Using \cref{alg:SGD} directly might require less training time, but identifying a support subsample may be very challenging.}

We further require the following assumption, which is verified if the calmness constraint qualification is satisfied and $c_1$ in \cref{eq:RCS_problem_tractable_unc} is chosen large enough.

\begin{assumption}\label{ass:RCS_feasibility_on_seen_samples}
For all $\delta\in\mathcal{S}$, $H_x x_t^\delta(\tilde{\theta}^*)\leq h_x$, where $\tilde{\theta}^*$ is obtained with \cref{alg:RCS_bpmpc_constraints}.
\end{assumption}

\begin{theorem}[{\cite[Theorem 1]{campi2018general}}]\label{thm:RCS_scenario}
Let Assumption \ref{ass:RCS_feasibility_on_seen_samples} hold, and let $\beta\in(0,1)$. Let $\epsilon:\{ 0,M \}\to[0,1]$ be defined as $\epsilon(M)=1$ and $\epsilon(k)=1-\sqrt[M-k]{\beta/(M\binom{M}{k})}$ for $k < M$. Then $\mathbb{P}_\mathcal{S} \{ V(\tilde{\theta}^*) > \epsilon(k^*) \} \leq \beta$, where \rz{$k^*=|\mathcal{T}^*|$}.
\end{theorem}


\cref{fig:RCS_summary} summarizes the complete algorithmic procedure proposed in this paper. \rzz{Choosing an appropriate size for $\mathcal{S}$ remains an open research question. In general, increasing the number of samples leads to tighter out-of-sample bounds at the expense of greater computational complexity. If $\mathcal{S}$ is too large, the resulting solutions may become overly conservative, depending also on the chosen MPC parameterization used in \cref{eq:PF_mpc}. We further study this tradeoff in the numerical example presented in \cref{subsec:SIM_linear}}
\begin{figure}
\begin{center}
\scalebox{0.95}{
\begin{tikzpicture}[auto, -latex, semithick, font=\scriptsize]

\node [align=center] (in) {Problem\\Specification};
\node [draw,rectangle,minimum height=2em,minimum width=2em,right of = in, node distance = 2cm, align=center] (alg2) {Optimize\\ Nominal\\ Performance};
\node [above of = alg2,node distance = 0.75cm] (text_alg2) {\hyperref[alg:NP_bpmpc_nominal]{(Alg. 3)}};
\node [draw,rectangle,minimum height=2em,minimum width=2em,right of = alg2, node distance = 2.15cm, align=center] (alg3) {Optimize\\ Constraint\\ Satisfaction};
\node [above of = alg3, node distance = 0.75cm] (text_alg3) {\hyperref[alg:RCS_bpmpc_constraints]{(Alg. 4)}};
\node [right of = alg3, node distance = 2.2cm, yshift = 0.5cm, align=center] (out1) {Deploy on\\ Real System};
\node [right of = alg3, node distance = 2.2cm, yshift = -0.5cm, align=center] (out2) {Compute\\ Confidence\\ (\cref{thm:RCS_scenario})};
\node [right of = out2, node distance = 1.4cm] (eps) {$\epsilon$};

\node (s) [below of = alg3, node distance = 1.15cm] {Data set $\mathcal{S}$};

\draw [draw,-latex] (in) -- (alg2);
\draw [draw,-latex] (alg2) -- node [pos=0.5,above] {$\theta^*$} (alg3);
\draw [draw,-latex] (alg2) -- node [pos=0.5,above] {$\theta^*$} (alg3);
\draw [draw,-latex] ([yshift=0.25cm]alg3.east) -- node [pos=0.5,above] {$\tilde{\theta}^*$} (out1);
\draw [draw,-latex] ([yshift=-0.25cm]alg3.east) -- node [pos=0.5,above] {$\mathcal{T}^*$} (out2);
\draw [draw,-latex] (out2.east) -- (eps);

\draw [draw,-latex] (s.north) -- (alg3.south);

\end{tikzpicture}}
\end{center}
\vspace{-0.3cm}
\captionsetup{belowskip=-10pt}
\caption{Summary of the entire algorithm.}\label{fig:RCS_summary}
\end{figure}

\section{Simulation example}
\label{section:SIM}
\subsection{Cart pendulum example}\label{subsec:SIM_nonlinear}

We test our method on the pendulum on a cart of \cite{guemghar2002predictive}, whose state is $(x,\dot{x},\phi,\dot{\phi})$ and whose dynamics are given by
\begin{align}\label{eq:SIM_cartpole}
\begin{split}
\ddot{\phi}(t) &= \frac{m\mu g \sin(\phi) - \mu \cos(\phi)(u + \mu \dot{\phi}^2\sin(\phi))}{mJ-\mu^2\cos(\phi)^2},\\
\ddot{x}(t) &= \frac{J(u+\mu \dot{\phi}^2\sin(\phi))-\mu^2g\sin(\phi)\cos(\phi)}{mJ-\mu^2\cos(\phi)^2},
\end{split}
\end{align}
where $x$ and $\dot{x}$ are the linear position and velocity of the cart, and $\phi$ and $\dot{\phi}$ are the angular position and velocity of the pendulum, respectively.
The input $u$ is the force applied to the cart.
We use Runge-Kutta $4$ with a sample time of $0.05$ seconds to obtain discrete time dynamics, and impose the constraints $|u(t)|\leq 0.75$, $|\phi(t)|\leq 0.2$, $|\dot{x}(t)| \leq 0.8$.
Note that the constraint on the angle $\phi$ requires the pendulum to remain near the upright position: this is challenging to satisfy, as the cart must move quickly to reach the origin, but not too quickly to avoid violating the constraint.
To retain definability, we reduce the domain of the trigonometric functions to a finite interval, and set the functions to zero outside.
The mass $m$, the inertia $J$ and the coefficient $\mu$ of the system are given by $m=\bar{m}(1+d_1)$, $J=\bar{J}(1+d_2)$, $\mu=\bar{\mu}(1+d_3)$, where $\bar{m}$, $\bar{J}$, $\bar{\mu}$ are known nominal values of \cite{guemghar2002predictive}, and $d=(d_1,d_2,d_3)$ is a random variable uniformly distributed in the set $[-0.05,0.05]^3$.
The noise $w_t$ is sampled uniformly from the set $\{ 0 \} \times [-0.01,0.01] \times \{ 0 \} \times [-0.1,0.1]$.
The initial condition is $x_0 = \bar{x}_0 + (0,\omega_1,0,\omega_2)$, with $\bar{x}_0=(-3,0,0,0)$ and $\omega_1,\omega_2$ are sampled independently and uniformly from the interval $[-0.3,0.3]$; note that in all cases only the velocity and angular velocity are affected by the uncertainty.
We use the linearization scheme described in \cite{zuliani2023bp} to obtain linear dynamics for \cref{eq:PF_mpc}.
Moreover, we choose a short horizon $N=5$ (whereas $T=120$) and $Q_x=\operatorname{diag}(1,0.001,1,0.001)$.

After running \cref{alg:NP_bpmpc_nominal} to obtain $\theta^*$, we run \cref{alg:RCS_bpmpc_constraints} with a set $\mathcal{S}$ of $1000$ samples (where each problem is solved with $1000$ GD iterations) choosing $c_1=40$, and $\alpha_k=0.1/k^{0.6}$.
We additionally add a squared $2$-norm penalty on the constraint violation to the cost of \cref{eq:RCS_problem_tractable_unc}, multiplied by the factor $c_2=40$.
This introduces an additional degree of freedom for tuning the algorithm without compromising the results of \cref{thm:RCS_bpmpc_constraints}.
After \cref{alg:RCS_bpmpc_constraints} terminates, we obtain $|\mathcal{T}^*|=3$, which provides a theoretical bound of $V_\mathcal{S}(\tilde{\theta}^*)\leq 0.0407$ with confidence $\beta=10^{-6}$.

For cross-validation, we test the tuned policy $\mpc(\cdot,\tilde{\theta}^*)$ on $1000$ unseen samples of $(w,d,x_0)$.
\cref{table:cost_and_violation} compares the performance of our method against the nominal MPC with $\theta^*$ obtained by \cref{alg:NP_bpmpc_nominal}, and against a nonlinear MPC controller utilizing the nominal nonlinear model for its dynamics and terminal state cost equal to the stage cost. To improve the performance of the nonlinear MPC, we increase its prediction horizon to $N=15$.
Both these alternatives fail to satisfy the constraints on every unseen scenario, whereas our method does not violate any constraints.

\begin{table}[ht]
\centering
\begin{tabular}{c c c c c}
\toprule
& \textbf{Average cost} & \multicolumn{3}{c}{\textbf{Violation}} \\
\cmidrule(lr){3-5}
& & Ratio & Total & Relative \\
\cmidrule(lr){1-5}
$\mpc(\cdot,\tilde{\theta}^*)$ & $330.163$ & $0\%$ & $0$ & $0\%$ \\
$\mpc(\cdot,\theta^*)$ & $293.482$ & $100\%$ & $0.646$ & $15.984\%$ \\
Nonlinear MPC & $299.992$ & $100\%$ & $10.772$ & $111.259\%$\\
\bottomrule
\end{tabular}
\caption{Closed-loop cost and constraint violation, cart pendulum}\label{table:cost_and_violation}
\end{table}

Figure \ref{fig:SIM_trajectories} shows the averaged closed loop trajectories (solid line) and the range spanned by $1000$ trajectories (shaded) of the linear position and velocity over time for the nominal MPC $\theta^*$ and the tuned MPC $\tilde{\theta}^*$. Note how the nominal MPC (in orange) is more aggressive in the earlier time-steps and therefore fails to guarantee constraint satisfaction under disturbances. The tuned MPC, on the other hand, manages to reduce the speed of the cart just enough to ensure safety.

\begin{figure}[ht]
\ifarxiv
    \centering
    \includegraphics{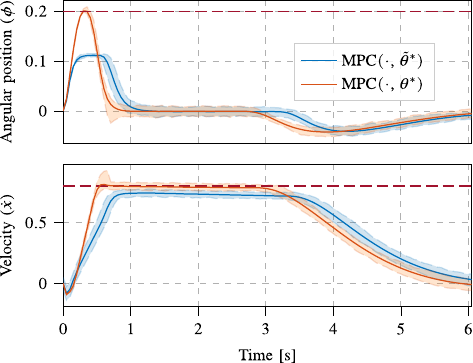}
\else
    \input{Figures/trajectories.tex}
    \captionsetup{belowskip=-8pt}
\fi
\caption{Average (solid line) and range (shaded area) of $1000$ closed-loop trajectories of the nominal and the tuned MPC schemes. The dashed red line represents the state constraints.}
\label{fig:SIM_trajectories}
\end{figure}

\rzz{Figure \ref{fig:cartpend_animation} shows the cart-pendulum system at different time steps under the control of the tuned MPC (top) and the nominal MPC (bottom), for a randomly generated unseen sample. While the nominal MPC brings the system to the origin more quickly (compare the two carts at $t=60$), it is also more aggressive during the initial phase of the motion, resulting in a violation of the angle constraint at $t=6$, as highlighted in the zoomed inset.}

\begin{figure}[ht]
\centering
\ifarxiv
    \includegraphics{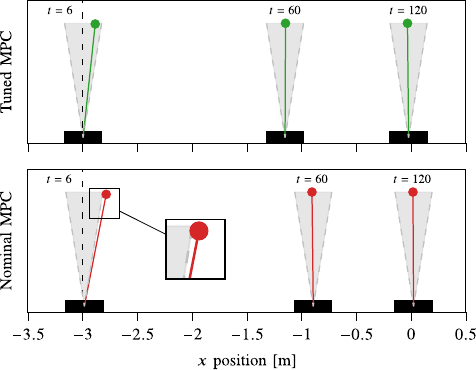}
\else
    \input{Figures/cartpend.tex}
\fi
\caption{Cart-pendulum under the action of the tuned MPC (top) and nominal MPC (bottom). Observe the violation of the angle constraint at time-step $t=6$ (in the inset).} \label{fig:cartpend_animation}
\end{figure}


\subsection{Quadrotor example}

Next, we compare with the nonlinear tube MPC of \cite{kohler2019computationally}, which we were not able to use in \cref{subsec:SIM_nonlinear} as the LMI (44) in \cite{kohler2019computationally}, necessary to obtain the shape of the tubes, was infeasible. When simulating our scheme, we consider the system dynamics and constraints of \cite{kohler2019computationally}, with the difference that the constraints on pitch and roll are tightened to $|\phi_i|\leq \pi/6$. We consider a larger additive noise $\|w\|^2\leq 2$ and uncertainty on $n_0$ and $d_0$ spanning the $\pm 20\%$ range. When simulating the robust MPC of \cite{kohler2019computationally}, however, we reduced the uncertainty to the range $\pm 11.5\%$ and removed the noise, as this was required to ensure the feasibility of the scheme.

The simulation results on $1000$ unseen samples are shown in Table \ref{table:cost_and_violation_quadrotor}. Note that, despite operating with no noise and less uncertainty, the scheme in \cite{kohler2019computationally} performs $6\%$ worse than our scheme, whereas both satisfied the constraints on all samples. Note that the performance of the nominal scheme $\theta^*$ is not much different from that of the robust scheme $\tilde{\theta}^*$, indicating that this problem may not be particularly challenging, even after increasing the uncertainty range, the noise magnitude, and reducing the constraints from \cite{kohler2019computationally}.
In this case, the P2L algorithm here terminated with $|\mathcal{T}^*|=2$, guaranteeing with confidence $1-10^{-6}$ that $V_{\mathcal{S}}(\tilde{\theta}^*)\leq 0.034$.

\begin{table}[ht]
\centering
\begin{tabular}{c c c c c}
\toprule
& \textbf{Average cost} & \multicolumn{3}{c}{\textbf{Violation}} \\
\cmidrule(lr){3-5}
&& Ratio & Total & Relative \\
\cmidrule(lr){1-5}
$\mpc(\cdot,\tilde{\theta}^*)$ & $283.590$ & $0\%$ & $0$ & $0$ \\
$\mpc(\cdot,\theta^*)$ & $282.536$ & $100\%$ & $0.392$ & $2.18\%$ \\
MPC of \cite{kohler2019computationally} & $300.604$ & $0\%$ & $0$ & $0$ \\
\bottomrule
\end{tabular}
\captionsetup{belowskip=-5pt}
\caption{Closed-loop cost and constraint violation, quadrotor}\label{table:cost_and_violation_quadrotor}
\end{table}

\subsection{Linear example} \label{subsec:SIM_linear}

We further compare our method to the Tube MPC of \cite{mayne2005robust} on the linear system of \cite{mayne2005robust}. In this setting, the model is known, but subject to additive noise. Following \cite{mayne2005robust}, we use $N=15$ for the tube MPC, but choose a shorter horizon $N=5$ for our scheme to better highlight its superior performance. By utilizing the same procedure described in Section \ref{subsec:SIM_nonlinear}, we draw $500$ samples to construct $\mathcal{S}$ and obtain $V_{\mathcal{S}}(\tilde{\theta}^*)\leq0.063$ with confidence $1-10^{-6}$ ($|\mathcal{T}^*|=2$). The results (obtained by simulating with $1000$ unseen samples) are reported in \cref{table:cost_and_violation_linear}. Note that we achieve a $50\%$ \eb{closed-loop} performance improvement without ever violating the constraints, \eb{highlighting the utility of the proposed framework in reducing conservatism}.

\begin{table}[ht]
\centering
\begin{tabular}{c c c}
\toprule
& \textbf{Average cost} \!\!\!\!\!\! & \textbf{Total Violation} \\
\cmidrule(lr){1-3}
$\mpc(\cdot,\tilde{\theta}^*)$ \!\!\!\!\!\!\! & $215.645$ & $0$ \\
Tube MPC of \cite{mayne2005robust} \!\!\!\!\!\!\! & $435.124$ & $0$ \\
\bottomrule
\end{tabular}
\captionsetup{belowskip=-5pt}
\caption{Closed-loop cost and constraint violation, linear}\label{table:cost_and_violation_linear}
\end{table}

\rzz{We now investigate how performance and constraint satisfaction vary with the size of the sample set $\mathcal{S}$. To this end, we generate $100$ random samples and order them in the set $\mathcal{S} = \{ \delta_{1}, \dots, \delta_{100} \}$ in increasing order of the constraint violation $\gamma^{\delta_i} = \sum_{t=0}^{T} \gamma(x_t^{\delta_i}(\tilde{\theta}^k))$ obtained with the nominal parameter $\theta^*$, so that $\gamma^{\delta_1} \leq \gamma^{\delta_2} \leq \dots \leq \gamma^{\delta_{100}}$. We then partition $\mathcal{S}$ into $20$ sets $\{\mathcal{S}^i\}_{i=1}^{20}$, each containing the first $5i$ samples of $\mathcal{S}$, and run \cref{alg:RCS_bpmpc_constraints} with $\mathcal{S}$ replaced with $\mathcal{S}^i$ for each $i$. Since the samples are ordered by their constraint violation, we expect the solution obtained with $\mathcal{S}^i$ to yield worse performance, but higher constraint satisfaction, on unseen data compared to the solution obtained with $\mathcal{S}^j$, whenever $i > j$.}
\begin{figure}[ht]
\centering
\ifarxiv
    \includegraphics{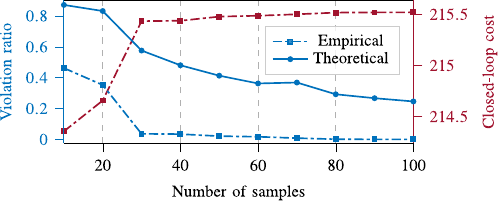}
\else
    \input{Figures/sample_size.tex}
    \captionsetup{belowskip=-5pt,aboveskip=-5pt}
\fi
\caption{Empirical and theoretical constraint violation chance (blue), and closed loop cost (red) as a function of the number of samples in $\mathcal{S}$.} \label{fig:samples}
\end{figure}

\rzz{\cref{fig:samples} confirms our hypothesis: observe how the impact on the performance is minimal, despite exhibiting a clear upward trend, whereas the empirical constraint violation chance, measured on $1000$ unseen samples, drops significantly from nearly $50\%$ to $0\%$. The theoretical chance of constraint violation, computed through \cref{thm:RCS_scenario}, also decreases while upper-bounding the empirical one. However, the theoretical bound need not be stricly decreasing: for example, a small increase occurs with $70$ samples, due to an increased number of support constraints. This phenomen is inevitable due to the a-posteriori nature of the P2L procedure. Future work may focus on how to obtain prior bounds that only depend on fixed problem quantities such as dimensions or number of samples.}

\subsection{On the effect of the penalty parameter} \label{subsection:penalty_example}
\rzz{We further study the effect of the penalty parameter $c_1$ on the closed-loop constraint violation using the same simulation example as in \cref{subsec:SIM_linear} with a single randomly generated sample $\delta$ extracted from $\P_w^T \times \P_d \times \P_{x_0}$. \cref{fig:penalty} reports the constraint violation across different iterations for increasing values of $c_1$. As expected, a larger value of the penalty leads to a faster decrease of the constraint violation. For small values of $c_1$, (approximately $c_1 \leq 2$) there is no improvement in the constraint violation within the simulated $1000$ iterations.}
\begin{figure}[ht]
\centering
\ifarxiv
    \includegraphics{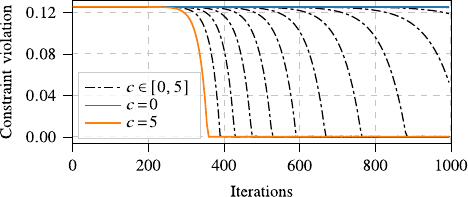}
\else
    \input{Figures/penalty.tex}
\fi
\caption{Closed-loop constraint violation across iterations for different values of penalty parameters.}\label{fig:penalty}
\end{figure}


\section{Conclusion and Limitations}
We proposed a principled way to design the cost and the constraint tightenings of an MPC scheme to improve good closed-loop performance and constraint satisfaction under noise and uncertainty on nonlinear systems.
We used the scenario approach to provide a probabilistic bound on the closed-loop constraint violation.
The resulting MPC problem is a QP that be solved efficiently with specialized software.

\rzz{Our approach can offer significant performance improvement while maintaining safety as long as the closed-loop trajectory of the true system (i.e., the one accounting for uncertainty and noise) does not deviate significantly from the nominal trajectory. However, the performance benefits may diminish if the set of possible initial conditions or the support of the noise is large. In such cases, we believe that traditional tube MPC methods (like the one proposed in \cite{kohler2019computationally}) may be a more appropriate choice.}

Future work may focus on developing a design strategy where the constraint violation chance is user-defined, or on improving the sample-efficiency of the algorithm.

\addtolength{\textheight}{-12.5cm}
\bibliographystyle{IEEEtran}
\bibliography{Sources/ref.bib}

\end{document}